%% file: LeGallSpaceLH-ArXiv-v3.tex
\documentclass[11pt,letterpaper]{article}
\usepackage{hyperref}
\input{macros}

\input{Refstyle.tex}
\begin{document}
\title{Classical Algorithms for Constant Approximation of the Ground State Energy of Local Hamiltonians}
\date{}
\author{
Fran{\c c}ois Le Gall\\
Graduate School of Mathematics\\
Nagoya University\\
legall@math.nagoya-u.ac.jp
}
\maketitle
\thispagestyle{empty}
\begin{abstract}
	We construct classical algorithms computing an approximation of the ground state energy of an arbitrary $k$-local Hamiltonian acting on $n$ qubits.
	
	We first consider the setting where a good ``guiding state'' is available, which is the main setting where quantum algorithms are expected to achieve an exponential speedup over classical methods. We show that a \emph{constant} approximation (i.e., an approximation with \emph{constant relative accuracy}) of the ground state energy can be computed classically in $\poly\left(1/\chi,n\right)$ time and $\poly(n)$ space, where~$\chi$ denotes the overlap between the guiding state and the ground state (as in prior works in dequantization, we assume sample-and-query access to the guiding state). This gives a significant improvement over the recent classical algorithm by Gharibian and Le Gall (SICOMP 2023), and matches (up to a polynomial overhead) both the time and space complexities of quantum algorithms for constant approximation of the ground state energy. We also obtain classical algorithms for higher-precision approximation.

	For the setting where no guided state is given (i.e., the standard version of the local Hamiltonian problem), we obtain a classical algorithm computing a constant approximation of the ground state energy in $2^{O(n)}$ time and $\poly(n)$ space. To our knowledge, before this work it was unknown how to classically achieve these bounds simultaneously, even for constant approximation. We also discuss complexity-theoretic aspects of our results. 
\end{abstract}
\newpage

\section{Introduction}\label{sec:intro}
\subsection{Statement of our main results}
Estimating the ground state energy of Hamiltonians is a central problem in both many-body physics and quantum complexity theory. 
Consider a $k$-local Hamiltonian 
\begin{equation}\label{eq:H}
H=
\sum_{i=1}^m H_i
\end{equation}
acting on $n$ qubits, with $k=O(1)$. Here each term $H_i$ acts non-trivially on only $k$ qubits (but does not need to obey any geometric locality). Let $\en{H}$ denote the ground state energy of $H$, i.e., its smallest eigenvalue. For any $\varepsilon>0$, we say that an estimate $\hat E$ is an $\varepsilon$-approximation of $\en{H}$ if 
\[
	\abs{\hat E-\en{H}}\le\varepsilon\sum_{i=1}^m\norm{H_i}.
\]
It is well known that computing a $1/\poly(n)$-approximation of $\en{H}$ is $\QMA$-hard, even for $k=2$ and for geometrically local Hamiltonians \cite{Cubitt+16,Kempe+06,Kitaev2002,Oliveira+08,Piddock+17}. The Quantum PCP conjecture \cite{Aharonov2013,Aharonov+02} posits that there exists a constant $\varepsilon>0$ such that computing an $\varepsilon$-approximation of $\en{H}$ is $\QMA$-hard as well.

Despite these hardness results, efficient quantum algorithms for ground state energy estimation can be constructed when a good ``guiding state'' is available, i.e., when a quantum state $\ket{\psi}$ that has a good overlap $\abs{\braket{\psi}{\psi_0}}$ with a ground state $\ket{\psi_0}$ of $H$ is given as an additional input or can be constructed easily (this problem has been called the ``guided local Hamiltonian problem'' in the recent literature \cite{Cade+ICALP23,Gharibian+SICOMP2023,Weggeman+24}). More precisely, quantum phase estimation \cite{Kitaev95,NC10} and more advanced techniques \cite{Apeldoorn+20,Ge+19,Gilyen+STOC19,Kerzner+23,Lin+20,Low+FOCS24,Martyn+21,Poulin+09} lead to the following result:
\begin{fact}\label{fact}
	Given a quantum state with overlap $\chi$ with a ground state of $H$, there exists a quantum algorithm that computes with high probability an $\varepsilon$-approximation of $\en{H}$ in $\poly\big(\frac{1}{\chi},\frac{1}{\varepsilon},n\big)$ time
	and $O\big(n+\log(\frac{1}{\varepsilon})\big)$ space.
\end{fact}

\noindent When $\chi=1/\poly(n)$ and $\varepsilon=1/\poly(n)$, both the running time and the space complexity (i.e., the number of bits and qubits needed for the computation) are polynomial in $n$. Even for larger values of $\chi$, the performance of this quantum algorithm can be significantly better than the performance of classical algorithms (which typically have running time exponential in $n$ --- see later for a detailed discussion).
Combined with the fact that for several important applications (e.g., quantum chemistry) good candidates for guiding states can be efficiently constructed, ground state energy estimation is one of the most promising and most anticipated applications of quantum computers (we refer to, e.g., \cite{Aaronson09, Abrams+99, Aspuru-Guzik+05, Bauer+20, Lee+21, Lee+23,Reiher+17} for discussions of these applications). 

In this work we investigate the classical complexity of this guided local Hamiltonian problem. A first issue is how to present the guiding state (which is a quantum state, i.e., an exponential-dimension vector) to a classical computer. 
As in prior works in dequantization \cite{Bakshi+SODA24,Chia2020a,Chia+JACM22,Du+20,Gharibian+SICOMP2023,Gilyen+22,Jethwani+MFCS20,LeGall23,TangSTOC19,Tang21}, we consider sample-and-query access: 
\begin{itemize}
	\item[(i)] for any $j\in[2^n]$ we can efficiently compute $\braket{j}{\psi}$\,;
	\item[(ii)] we can efficiently sample from the probability distribution $p\colon[2^n]\to[0,1]$ that outputs $j$ with probability $\abs{\braket{j}{\psi}}^2$.
\end{itemize}
The motivation for (ii), which is the central assumption in dequantized algorithms,
is as follows: since measuring the quantum state $\ket{\psi}$ in the computational basis gives a sample from the probability $p$, it is natural (or ``fair'') to assume that in the classical setting this distribution is efficiently samplable as well.

Recently, Gharibian and Le Gall \cite{Gharibian+SICOMP2023} constructed a classical algorithm computing an $\varepsilon$-approximation of $\en{H}$ in $n^{O(\log(1/\chi)/\varepsilon)}$ time by dequantizing quantum algorithms based on the Quantum Singular Value Transformation. 
Here is our main result:

\begin{theorem}[Simplified version]\label{th:main-informal}
	Given sample-and-query access to a quantum state with overlap $\chi$ with a ground state of $H$, there exists a classical algorithm that computes with high probability an $\varepsilon$-approxima\-tion of $\en{H}$ in $\poly\big(\frac{1}{\chi^{1/\varepsilon}},n\big)$ time and $\poly\big(n,\frac{1}{\varepsilon}\big)$ space.
\end{theorem}

Our result significantly improves the running time of the algorithm from \cite{Gharibian+SICOMP2023}. 
For instance, if $\chi=\Omega(1)$, i.e., if we have a good guiding state, our algorithm has time complexity $\poly(2^{1/\varepsilon},n)$ instead of $n^{O(1/\varepsilon)}$ in \cite{Gharibian+SICOMP2023}. If $\varepsilon=\Omega(1)$, i.e., if we want only constant precision, our algorithm has time complexity $\poly(1/\chi,n)$ instead of $n^{O(\log(1/\chi))}$ in \cite{Gharibian+SICOMP2023}. 
Additionally, our approach only uses polynomial space. Comparing \cref{th:main-informal} with the bounds of Fact \ref{fact} shows that for \emph{constant} precision, there exist a classical algorithm matching (up to a polynomial overhead) the performance of quantum algorithms. 

Using the same technique, we also obtain the following result for the case where no guided state is given (i.e., the standard version of the local Hamiltonian problem):

\begin{theorem}[Simplified version]\label{cor:main-informal}
	For any constant $\varepsilon>0$, there exists a classical algorithm that computes with high probability an $\varepsilon$-approximation of $\en{H}$ in $2^{O(n)}$ time and $\poly(n)$ space.
\end{theorem}
To our knowledge, before this work it was unknown how to achieve simultaneously running time $2^{O(n)}$ and space complexity $\poly(n)$ for ground state energy estimation of arbitrary local Hamiltonians (even for constant $\varepsilon$). We will further discuss the implications of our results in Section \ref{sub:impl} after reviewing known classical approaches for ground state energy estimation in the next subsection.

\subsection{Background on classical approaches for ground state energy estimation}\label{sub:rev}
There are two main classical approaches for estimating the ground state energy of a local Hamiltonian:
\begin{itemize}
	\item The power method or its variant the Lanczos method \cite{KW92,Lanczos50}, which estimates the ground state using matrix-vector multiplications. Since the Hamiltonian is a (sparse) matrix of dimension $2^n$, the time complexity is $\myO{2^n}$.\footnote{In this paper the notation $\myO{\cdot}$ suppresses the $\poly(n)$ factors.} There are two main issues with this approach. First, it requires storing explicit vectors in memory, which leads to space complexity $\Omega(2^n)$ and significantly reduces its applicability. Second, it is unclear how a guiding state would help significantly reduce the time complexity (having a good guiding state does reduce the number of iterations, but each iteration still requires matrix-vector multiplications of matrices and vectors of dimension $2^n$).
	\item Quantum Monte Carlo methods, which use sampling arguments to estimate the ground state without having to store explicitly the quantum state. This approach is especially useful for ``stoquastic'' Hamiltonians, i.e., Hamiltonians for which all the off-diagonal elements are real and non-positive, and has lead to the design of classical algorithms as well as 
	as complexity-theoretic investigations of the complexity of the local Hamiltonian problem for stochastic Hamiltonians \cite{Bravyi+08,Bravyi+10,Bravyi15,Liu2021}. While some of these techniques have been extended to a few classes of non-stoquastic local Hamiltonians, such as gapped local Hamiltonians \cite{Bravyi+23} or arbitrary Hamiltonians with succinct ground state \cite{Jiang25}, for ground state energy estimation the ``sign-problem'' significantly limits its applications to arbitrary local Hamiltonians \cite{Hangleiter+20,Troyer+05}. It is also unclear how the guiding state would help reduce the time complexity.
\end{itemize}

A third approach 
is direct classical simulation of the quantum circuit used in Fact \ref{fact}. 
There are several techniques for simulating quantum circuits on a classical computer. If the circuit acts on~$n$ qubits and has $m$ gates, the Schr\"odinger method stores the entire state vector in memory and performs successive matrix-vector multiplications, using roughly $m2^n$ time and $2^n$ space. While the space complexity can in several cases be significantly reduced using matrix product states or more general tensor networks \cite{Biamonte+17}, those representations also require exponential space in the worst case. On the other hand, the Feynman method calculates an amplitude as a sum of terms, using roughly $4^{m}$ time and $m+n$ space (this approach was used by Bernstein and Vazirani \cite{Bernstein+97} to prove the inclusion $\BQP\subseteq\mathsf{P}^{\#\mathsf{P}}$). Aaronson and Chen \cite{Aaronson+CCC17} have introduced a recursive version of the Feynman method, inspired by the proof of Savitch theorem \cite{Savitch70}, that works in $2^{O(n\log m)}$ time and $\poly(n,m)$ space.\vspace{-2mm}

\paragraph{Other prior works.}
Several ground state energy estimation classical algorithms have also been developed for special classes of Hamiltonians, such as one-dimensional gapped local Hamiltonians \cite{Landau+15}, quantum analogues of Max Cut \cite{Anshu+TQC20,Hallgren+20,King23,Lee+ICALP24,Parekh+ICALP21}, or Hamiltonians defined on structured graphs \cite{Bansal+09,Bergamaschi23,Brandao+STOC13}. Additionally, there are a few works \cite{Bravyi+19,Gharibian+SICOMP12,Harrow+17} achieving weaker (but still nontrivial) approximation ratios of the ground state energy, and a recent work \cite{Buhrman+PRL25} achieving a constant approximation ratio for any local Hamiltonian in time slightly better than $2^n$ (but not space-efficiently). 

\subsection{Implication of our results}\label{sub:impl}
We now discuss several implications of our results.\vspace{-2mm}

\paragraph{Better understanding of the quantum advantage.}
As already mentioned, Theorem \ref{th:main-informal} implies that for any constant precision parameter~$\varepsilon$, we can construct classical ground state energy estimation algorithms with performance matching (up to a polynomial overhead) the performance of the best known quantum algorithms. While Ref.~\cite{Gharibian+SICOMP2023} already showed this for $\chi=O(1)$, \cref{th:main-informal} proves this result for any value~$\chi\in(0,1]$. This implies that (under the assumption that current quantum algorithms are optimal) there is no superpolynomial quantum advantage for the constant-precision guided local Hamiltonian problem and gives another strong evidence that exponential quantum advantage for ground state energy estimation (and applications to, e.g., quantum chemistry) comes from the improved precision achievable in the quantum setting.\vspace{-2mm}

\paragraph{Space-efficient classical ground state energy estimation algorithms.} The second main contribution of this work is the design of space-efficient algorithms for ground state energy estimation. In particular, for the case where no guiding state is available, \cref{cor:main-informal} gives a $2^{O(n)}$-time $\poly(n)$-space classical algorithm. As already mentioned, to our knowledge
before this work it was unknown how to achieve simultaneously running time $2^{O(n)}$ time and space complexity $\poly(n)$ for arbitrary local Hamiltonians: even for constant $\varepsilon$, the best running time was $2^{O(n\log n)}$ by the approach by Aaronson and Chen~\cite{Aaronson+CCC17}. \vspace{-2mm}

\paragraph{Potentially practical classical ground state energy estimation algorithms.}

From a practical perspective, the potential of \cref{th:main-informal} is even more striking. 
In particular, for $\chi=\Omega(1)$, i.e., when we have access to a fairly good guided state (which is the case in some applications to quantum chemistry), we obtain a  $\poly(2^{1/\varepsilon},n)$-time $\poly(n)$-space classical algorithm. Note that the running time is polynomial even for $\varepsilon=1/\log n$.
Additionally, the running time is better than~$2^n$, which is the typically running time of other classical methods, even for precision as low as $\varepsilon=c/n$, for a small enough constant $c>0$. While evaluating the practicality of our algorithm is beyond the scope of this paper, we hope our algorithms find applications in many-body physics.\vspace{-2mm}

\paragraph{Complexity-theoretic implications.}
In order to formally discuss complexity theoretic aspects of our results and their relations with standard complexity classes as $\BPP$, $\BQP$ and $\QMA$, we first need to introduce decision versions of our problems. As standard in Hamiltonian complexity theory, we add the promise that either (i) $\en{H}\le a$ or (ii) $\en{H}>b$ holds, for some values $a,b\in[0,1]$ such that $b-a>\varepsilon$, and ask to decide which of (i) or (ii) holds. This leads to the following (standard) decision version of the local Hamiltonian problem:
\begin{center}
	\ovalbox{
	\begin{minipage}{12.5 cm} \vspace{2mm}
	
	\noindent\hspace{3mm}$\LH(\varepsilon)$\hspace{5mm}{\tt(Local Hamiltonian problem } ---\:  {\tt  decision version)}\\\vspace{-3mm}
	
	\noindent\hspace{3mm} Input: $\bul$ a $O(1)$-local Hamiltonian $H$ as in \cref{eq:H} acting on $n$ qubits

	\noindent\hspace{15mm}
	$\bul$ two numbers $a,b\in[0,1]$ such that $b-a>\varepsilon$
	\vspace{2mm}

	\noindent\hspace{3mm} Promise: either (i) $\en{H}\le a$ or (ii) $\en{H}>b$ holds
	
	\vspace{2mm}
	\noindent\hspace{3mm} Goal: decide which of (i) or (ii) holds
	\vspace{2mm}
	\end{minipage}
	}\vspace{2mm}
	\end{center}

As already mentioned, the problem $\LH(\varepsilon)$ is $\QMA$-complete for $\varepsilon=1/\poly(n)$. On the other hand, for any $\varepsilon=O(1)$ \cref{cor:main-informal} leads to the inclusion
\[
	\LH(\varepsilon)\in\BPTS\left(2^{O(n)},\:\poly(n)\right), 
\]
where $\BPTS(t(n),s(n))$
denotes the class of (promise) decision problems that can be solved with probability at least $2/3$ by a probabilistic Turing machine running in $t(n)$ time and using $s(n)$ space. 

For the guided local Hamiltonian problem, another subtle issue is how to access the guiding state $\ket{\psi}$. So far we have assumed 
sample-and-query access to $\ket{\psi}$ when considering classical algorithms. While satisfactory when discussing algorithmic aspects of the problem (as we did so far), such ``oracle'' access to the input is problematic if we want to discuss relations with standard complexity classes such as $\BPP$, $\BQP$ and $\QMA$. Instead, we make the following assumptions: in the quantum setting, the description of a quantum polynomial-size circuit creating $\ket{\psi}$ is given as input; in the classical setting, the description of a classical polynomial-size circuit implementing sample-and-query access to $\ket{\psi}$ is given as input. This leads to the following decision version of the guided local Hamiltonian problem:

\begin{center}
	\ovalbox{
	\begin{minipage}{14.5 cm} \vspace{2mm}
	
	\noindent\hspace{3mm}$\GLH(\varepsilon,\chi)$\hspace{5mm}{\tt(Guided Local Hamiltonian problem } ---\:  {\tt  decision version)}\\\vspace{-3mm}
	
	\noindent\hspace{3mm} Input: $\bul$ a $O(1)$-local Hamiltonian $H$ as in \cref{eq:H} acting on $n$ qubits
	
	\noindent\hspace{15mm}
	$\bul$ the description of a $\poly(n)$-size circuit implementing access to a quantum  
	
	\hspace{25mm}
	state $\ket{\psi}$
	with overlap at least $\chi$ with the ground state of $H$

	\noindent\hspace{15mm}
	$\bul$ two numbers $a,b\in[0,1]$ such that $b-a>\varepsilon$
	\vspace{2mm}

	\noindent\hspace{3mm} Promise: either (i) $\en{H}\le a$ or (ii) $\en{H}>b$ holds
	
	\vspace{2mm}
	\noindent\hspace{3mm} Goal: decide which of (i) or (ii) holds
	\vspace{2mm}
	\end{minipage}
	}\vspace{2mm}
\end{center}

Prior results on the hardness of the guided local Hamiltonian problem~\cite{Cade+ICALP23,Gharibian+SICOMP2023} combined with \cref{fact} imply that $\GLH(\varepsilon,\chi)$ is $\BQP$-complete for $\varepsilon=1/\poly(n)$ and constant $\chi$. Ref.~\cite{Gharibian+SICOMP2023} also showed that this problem is in the class $\BPP$ for constant $\varepsilon$ and constant $\chi$. \cref{th:main-informal} enables us to strengthen this result and show that for constant $\varepsilon$, the inclusion
$
	\GLH(\varepsilon,\chi)\in\BPP
$
holds for $\chi=1/\poly(n)$ as well.

These complexity-theoretic implications are summarized in \cref{table:comp}.\vspace{-2mm}

\begin{table}[h]
	\centering
	\caption{The complexity of the problems $\LH(\varepsilon)$ and $\GLH(\varepsilon,\chi)$.} 
	\renewcommand{\arraystretch}{1.4}
	\label{table:comp}
	\begin{tabular}{ c|c||c|c| } 
		\cline{2-4}
		& $\LH(\varepsilon)$& \multicolumn{1}{c|}{$\GLH(\varepsilon,\Theta(1))$} & $\GLH\big(\varepsilon,\frac{1}{\poly(n)}\big)$
			\bigstrut\\ 
		\hline
		\multicolumn{1}{|l|}{$\varepsilon=\Theta(1)$} &
		in $\BPTS(2^{O(n)},\poly(n))$\:\:(Th.~\ref{cor:main-informal})& in $\BPP$\:\:(Ref.~\cite{Gharibian+SICOMP2023})& in $\BPP$\:\:(Th.~\ref{th:main-informal})
		\\\hline 
		\multicolumn{1}{|l|}{$\varepsilon=\frac{1}{\poly(n)}$}\bigstrut &$\QMA$-complete\:\:(Refs.~\cite{Kempe+06,Kitaev2002})& \multicolumn{2}{c|}{$\BQP$-complete\:\:(Refs.~\cite{Cade+ICALP23,Gharibian+SICOMP2023})}\\
		\hline
	\end{tabular}
\end{table}

\subsection{Technical overview}
We now give a technical overview of our results. We first describe the concept of sparse decomposition of matrices. We then present the three main techniques of our proof and explain how to combine them: eigenvalue estimation via polynomial transformations, iterated matrix multiplication via sample-and-query access and, finally, implementation of the polynomial transformation (our main contribution). \vspace{-2mm}

\paragraph{Sparse decomposition of matrices.}
For an integer $s\ge 0$, we say that a matrix is $s$-sparse if each row contains only at most $s$ nonzero entries.\footnote{In the literature, it is often required that each row and each column contains only $s$ nonzero entries. In this paper we only need the sparsity condition for rows.} 

In this work we introduce the concept of \emph{$(s,\upw)$-decomposition}. 
For a matrix $A$, an integer $s\ge 0$ and a real number $\upw\ge 0$, an $(s,\upw)$-decomposition of $A$ is a decomposition 
\[
	A=\sum_{i=1}^m A_i  
	\hspace{3mm}\textrm{with}\hspace{2mm}
	\sum_{i=1}^m \norm{A_i}\le \upw
\]
in which $A_i$ is an $s$-sparse matrix for each $i\in[m]$ (the complete definition is given in Section~\ref{sec:prelim}).
We develop algorithms for estimating the smallest eigenvalue of normal matrices with an $(s,\upw)$-decomposi\-tion. The parameter $s$ will control the complexity of the algorithm, while the parameter $\upw$ will control the accuracy of the algorithm (i.e., the precision of the estimator). 

Note that for a $k$-local Hamiltonian, the decomposition of \cref{eq:H} is precisely an $(s,\kappa)$-decompo\-sition with $s=2^k$ and $\upw=\sum_{i=1}^m\norm{H_i}$. Theorems \ref{th:main-informal} and \ref{cor:main-informal} will be obtained as corollaries of similar but more general result for arbitrary normal matrices with an $(s,\upw)$-decomposition. \vspace{-2mm}

\paragraph{Eigenvalue estimation via polynomial transformation.}
We will estimate the smallest eigenvalue using polynomial transformations. This is the same approach as the one used by quantum algorithms based on the Quantum Singular Value Transform for eigenvalue filtering (see for instance \cite[Appendix B.8]{Martyn+21} for a good overview) and their dequantized version~\cite{Gharibian+SICOMP2023}.

Consider a normal matrix $A\in \Comp^{2^n\times 2^n}$ that has an $(s,\upw)$-decomposition. In this overview we assume that $\upw=1$ (the general case can be reduced to this case by renormalizing the matrix). We write this decomposition as
\begin{equation}\label{eq:A}
	A=\sum_{i=1}^m A_i
	\hspace{3mm}\textrm{with}\hspace{2mm}
	\sum_{i=1}^m\norm{A_i} = 1,
\end{equation}
where each $A_i$ is an $s$-sparse matrix
(here we assume for simplicity that $\sum_{i=1}^m\norm{A_i} = 1$ instead of $\sum_{i=1}^m\norm{A_i} \le 1$). Since $A$ is normal, it can be written as 
\[
	A=\sum_{i=1}^{2^n}\lambda_i\ket{u_i}\bra{u_i},
\]
where
$\lambda_1\le \lambda_2\le\cdots\le\lambda_{2^n}$ are the eigenvalues of $A$ and $\ket{u_1},\ldots,\ket{u_{2^n}}$ are corresponding unit-norm eigenvectors. The smallest eigenvalue of $A$ is $\en{A}=\lambda_1$.
Since $\norm{A}\le 1$, we have $\abs{\lambda_i}\le 1$ for all $i\in[2^n]$. In this overview we will assume for simplicity that $\en{A}\ge 0$.

Estimating the smallest eigenvalue reduces (by binary search) to the following decision version: under the promise that either (i) $\en{A}\le a$ or (ii) $\en{A}>b$ holds, for some values $a,b\in[0,1]$ such that $b-a>\varepsilon$, decide which of (i) or (ii) holds. For simplicity, in this overview we assume that in case~(i) we have only one eigenvalue smaller than or equal to $a$, i.e., $\lambda_i>b$ for all $i\in\set{2}{2^n}$. 

The idea is to take a (low-degree) polynomial $P\in\Real[x]$ such that 
\begin{equation}\label{eq:Pintro}
	\begin{cases}
	P(x)\approx 1& \textrm{ if } x\in[-1,a],\\
	P(x)\approx 0& \textrm{ if } x\in[b,1].\\
	\end{cases}
\end{equation}
This is done by choosing a low-degree polynomial that approximates well the ``rectangle'' function.
From \cref{eq:Pintro} we have
\[
	P(A)=\sum_{i=1}^{2^n}P(\lambda_i)\ket{u_i}\bra{u_i}
	\approx
	\begin{cases}
		\ket{u_1}\bra{u_1}& \textrm{ if } \en{A}\in[0,a],\\
		0& \textrm{ if } \en{A}\in[b,1].\\
		\end{cases}
\]
For any state $\ket{\psi}$, this implies 
\[
\bra{\psi}P(A)\ket{\psi}
	\approx
	\begin{cases}
		\chi^2& \textrm{ if } \en{A}\in[0,a],\\
		0& \textrm{ if } \en{A}\in[b,1],\\
		\end{cases}
\]
where $\chi=\abs{\braket{u_1}{\psi}}$ is the overlap between $\ket{\psi}$ and the ground state $\ket{u_1}$. This means that to decide which of (i) or (ii) holds, it is enough to estimate the quantity $\bra{\psi}P(A)\ket{\psi}$.\vspace{-2mm}
\paragraph{Iterated matrix multiplication using sample-and-query access.}
Before explaining how we estimate $\bra{\psi}P(A)\ket{\psi}$, we discuss a related problem: estimating $\bra{\psi}B_1\cdots B_r\ket{\psi}$ for arbitrary $s$-sparse matrices $B_1,\ldots,B_r$. Using techniques from prior works on dequantization \cite{Chia+JACM22,LeGall23,TangSTOC19}, a good estimate of $\bra{\psi}B_1\cdots B_r\ket{\psi}$ can be efficiently computed given sample-and-query access to $\ket{\psi}$ and (efficient) query access to the vector $B_1\cdots B_r\ket{\psi}$. 

To efficiently implement query access to $B_{1}\cdots B_{r}\ket{\psi}$, we give a space-efficient version of an approach for iterated matrix multiplication used in \cite{Gharibian+SICOMP2023,Schwarz2013}. This approach is based on the following idea: to obtain the $\ell$-th entry of $B_{1}\cdots B_{r}\ket{\psi}$, we only need to know the $s$ nonzero entries of the $\ell$-th row of $B_1$, which can be queried directly, together with the corresponding entries in the vector $B_{2}\cdots B_{r}\ket{\varphi}$, which can be computed recursively. We show in Section \ref{sec:IMM} that the running time of this approach is $\myO{s^r}$ and its space complexity $\poly(n)$ when $r\le \poly(n)$.\vspace{-2mm}

\paragraph{Implementation of the polynomial transformation.}
In order to estimate $\bra{\psi}P(A)\ket{\psi}$, our key idea is to exploit the decomposition of $A$ as a sum of sparse matrices (\cref{eq:A}). 

Write $P(x)=a_0+a_1x+\cdots a_dx^d$. Since 
\[
	\bra{\psi}P(A)\ket{\psi}=\sum_{r=0}^d a_r \bra{\psi}A^r\ket{\psi},
\]
it is enough to compute $\bra{\psi}A^r\ket{\psi}$, for each $r\in\set{0}{d}$, with enough precision (the required precision depends on the coefficient~$a_r$ and the degree $d$). 

Our core technical contribution shows how to efficiently estimate $\bra{\psi}A^r\ket{\psi}$. The approach from \cite{Gharibian+SICOMP2023} directly estimated this quantity by using iterated matrix multiplication via sample-and-query access with 
\begin{equation}\label{eq:first}
	B_1=\cdots=B_r=A.
\end{equation}
Instead, we use the decomposition of $A$ to write 
\[
	\bra{\psi}A^r\ket{\psi}=\bra{\psi}\left(\sum_{i=1}^m A_i\right)^r\ket{\psi}=
	\sum_{x\in[m]^r}\bra{\psi}A_{x_1}\cdots A_{x_r}\ket{\psi},
\]
and show how to estimate this quantity
by sampling $x$ from the set $[m]^r$ according to an appropriate probability distribution and 
then estimating $\bra{\psi}A_{x_1}\cdots A_{x_r}\ket{\psi}$ by using iterated matrix multiplication via sample-and-query access with 
\begin{equation}\label{eq:second}
	B_1=A_{x_1},\ldots, B_r=A_{x_r}.
\end{equation}
The crucial point is that the matrices in \cref{eq:second} are significantly sparser than the matrices in \cref{eq:first}, which leads to our improvement of the running time. 
Indeed, while each term $A_{x_i}$ is $s$-sparse, the whole matrix $A$ is only $t$-sparse with $t=sm$. 
Using the matrices of \cref{eq:second} thus enables us to reduce the time complexity of the iterated matrix multiplication from $\myO{t^r}$ to $\myO{s^r}$. This improvement is especially remarkable for the case of local Hamiltonians:
for a $k$-local Hamiltonian with $k=O(1)$ and $m=\poly(n)$ terms, we have $s=O(1)$ and $t=\poly(n)$.

We now give more details about the sampling process. 
We first consider the probability distribution $q\colon[m]^r\to[0,1]$ defined 
as $q(x)=\norm{A_{x_1}}\cdots\norm{A_{x_r}}$ for any $x\in[m]^r$. Then we consider the random variable
\begin{equation}\label{eq:2}
	\frac{\bra{\psi}A_{x_1}\cdots A_{x_r}\ket{\psi}}{q(x)},
\end{equation}
where $x$ is sampled according to $q$.
It is easy to see that this is an unbiased estimator of $\bra{\psi}A^r\ket{\psi}$. While the exact value of $\bra{\psi}A_{x_1}\cdots A_{x_r}\ket{\psi}$ in \cref{eq:2} cannot be computed efficiently, we can use iterated matrix multiplication with the matrices of \cref{eq:second} to estimate $\bra{\psi}A_{x_1}\cdots A_{x_r}\ket{\psi}$, which leads to a biased estimator of $\bra{\psi}A^r\ket{\psi}$. We show that the bias can be controlled and show that we can  bound the variance as well. Taking the mean of a reasonably small number of repetitions thus gives a good estimation of $\bra{\psi}A^r\ket{\psi}$, as desired.




\paragraph{Generalized overlap.}
To cover a broader range of applications, we make one slight generalization to the notion of overlap presented in this introduction. 
Since we are only interested in computing an $\varepsilon$-approximation of the ground state energy, instead of considering the overlap of the guiding state with the ground state energy subspace, we can consider the overlap of the guiding state with the whole subspace of energy in the interval $[\en{A},(1+\Theta(\varepsilon))\en{A}]$. In most applications the latter can be significant larger than the former. We refer to Section \ref{sec:prelim} for more details.

\subsection{Organization of the paper}
After giving formal definitions and presenting some lemmas in Section \ref{sec:prelim}, we present the three techniques mentioned above in Sections \ref{sec:IMM}, \ref{sec:PT} and \ref{sub:EV1}. Our main technical contribution is the implementation of the polynomial transform (Section \ref{sec:PT}). We give the full statements and proofs of Theorems \ref{th:main-informal} and \ref{cor:main-informal} in Section \ref{sub:EV2}.

\input{Preliminaries.tex}
\input{IMM.tex}
\input{PT.tex}
\input{EV.tex}

 \section*{Acknowledgments}
 The author is supported by JSPS KAKENHI grant No.~24H00071, MEXT Q-LEAP grant No.~JPMXS0120319794, JST ASPIRE grant No.~JPMJAP2302 and JST CREST grant No.~JPMJCR24I4. 

\printbibliography

\appendix
\input{appendix.tex}
	
\end{document}

%% file: macros.tex
\usepackage[margin=1in]{geometry}
\usepackage{amsmath,amssymb,amsfonts,setspace,url}
\usepackage{bm}
\usepackage[boxed,linesnumbered,noend]{algorithm2e}
\DontPrintSemicolon
\let\oldnl\nl
\newcommand{\nonl}{\renewcommand{\nl}{\let\nl\oldnl}}
\usepackage{algpseudocode}
\usepackage{latexsym}
\usepackage{amsthm}
\usepackage{subfig}
\usepackage{float}
\usepackage{stmaryrd}
\usepackage{fancybox}
\usepackage{bigstrut,array,multirow}
\usepackage{here}
\usepackage{color}
\usepackage{url}
\usepackage[capitalise]{cleveref}
\usepackage{mathpazo}
\theoremstyle{plain}
\newtheorem{theorem}{Theorem}

\newtheorem{lemma}{Lemma}

\newtheorem{definition}{Definition}
\newtheorem{fact}{Fact}

\newtheorem{proposition}{Proposition}
\newtheorem{claim}{Claim}

\theoremstyle{definition}

\newcommand{\Nat}{\mathbb{N}}

\newcommand{\ket}[1]{|#1\rangle}
\newcommand{\braket}[2]{\langle#1|#2\rangle}

\newcommand{\norm}[1]{\lVert #1\rVert}

\newcommand{\abs}[1]{\left| #1\right|}
\newcommand{\bra}[1]{\langle#1|}

\newcommand{\upw}{\kappa}

\newcommand{\E}{\mathbb{E}}
\newcommand{\bul}{\ast}

\newcommand{\Tt}{\mathcal{T}}

\newcommand{\qt}[1]{\mathsf{qt}(#1)}

\newcommand{\en}[1]{\mathcal{E}\!\left(#1\right)}

\newcommand{\overlapf}[2]{\Gamma_{\!#2}(#1)}
\newcommand{\len}[1]{\mathsf{len}(#1)}
\newcommand{\qs}[1]{\mathsf{qs}(#1)}

\newcommand{\ceil}[1]{\left\lceil #1 \right\rceil}
\newcommand{\ex}[1]{\E\left[ #1 \right]}

\newcommand{\var}[1]{\mathbb{V}\left[ #1 \right]}

\newcommand{\pr}[1]{\mathrm{Pr}\left[ #1 \right]}

\newcommand{\Real}{\mathbb{R}}
\newcommand{\Comp}{\mathbb{C}}

\newcommand{\QMA}{\mathsf{QMA}}
\newcommand{\BPTS}{\mathsf{BPTimeSpace}}

\newcommand{\BQP}{\mathsf{BQP}}
\newcommand{\BPP}{\mathsf{BPP}}

\newcommand{\myO}[1]{O^\ast\!\left(#1\right)}
\newcommand{\MM}{\mathsf{IMM}}
\newcommand{\IP}{\mathsf{IP}}
\newcommand{\EST}{\mathsf{PT}}
\newcommand{\EV}{\mathsf{SE}}
\newcommand{\GLH}{\mathsf{GLH}}
\newcommand{\LH}{\mathsf{LH}}
\newcommand{\Test}{\mathsf{Test}}

\providecommand{\Ss}{\mathcal{S}}

\providecommand{\Aa}{\mathcal{A}}

\providecommand{\Comp}{\mathbb{C}}
\providecommand{\Bb}{\mathcal{B}}

\newcommand{\poly}{\mathrm{poly}}
\newcommand{\spa}{\mathrm{span}}

\newcommand{\set}[2]{\{#1,\ldots,#2\}}

%% file: Refstyle.tex
\usepackage[backend=bibtex, style=ext-alphabetic, maxbibnames=99]{biblatex}
\renewbibmacro{in:}{} 
\DeclareFieldFormat*{title}{#1} 
\DeclareFieldFormat{pages}{#1} 
\DeclareFieldFormat
  [article,inbook,incollection,inproceedings,patent,thesis,unpublished,misc]
  {titlecase:title}{\MakeSentenceCase*{#1}}

  \DeclareBibliographyDriver{inproceedings}{%
  \usebibmacro{bibindex}%
  \usebibmacro{begentry}%
  \usebibmacro{author/translator+others}%
  \setunit{\printdelim{nametitledelim}}\newblock
  \usebibmacro{title}%
  \newunit
  \printlist{language}%
  \newunit\newblock
  \usebibmacro{byauthor}%
  \newunit\newblock
  \usebibmacro{in:}%
  \usebibmacro{maintitle+booktitle}%
  \newunit\newblock
  \usebibmacro{event+venue+date}%
  \newunit\newblock
  \usebibmacro{byeditor+others}%
  \newunit\newblock
  \iffieldundef{maintitle}
    {\printfield{volume}%
     \printfield{part}}
    {}%
  \newunit
  \printfield{volumes}%
  \newunit\newblock
  \usebibmacro{series+number}%
  \newunit\newblock
  \printfield{note}%
  \newunit\newblock
  \printlist{organization}%
  \addcomma\addspace
  \usebibmacro{chapter+pages}%
  \addcomma\newblock
  \usebibmacro{publisher+location+date}%
  \newunit\newblock
  \iftoggle{bbx:isbn}
    {\printfield{isbn}}
    {}%
  \newunit\newblock
  \newunit\newblock
  \usebibmacro{addendum+pubstate}%
  \setunit{\bibpagerefpunct}\newblock
  \usebibmacro{pageref}%
  \newunit\newblock
  \iftoggle{bbx:related}
    {\usebibmacro{related:init}%
     \usebibmacro{related}}
    {}%
    \usebibmacro{doi+eprint+url}%
  \usebibmacro{finentry}}

  \newtoggle{bbx:datemissing}

\renewbibmacro*{date}{\toggletrue{bbx:datemissing}%
}

\renewbibmacro*{issue+date}{%
  \toggletrue{bbx:datemissing}%
  \iffieldundef{issue}
    {}
    {\printtext[parens]{%
       \printfield{issue}}}%
  \newunit}

\newbibmacro*{date:print}{%
  \togglefalse{bbx:datemissing}%
  \printdate}

\renewbibmacro*{note+pages}{%
  \printfield{note}%
  \setunit{\bibpagespunct}%
  \printfield{pages}%
  \setunit{\addcomma\space}
  \usebibmacro{date:print}%
  \newunit}

\renewbibmacro*{addendum+pubstate}{%
  \iftoggle{bbx:datemissing}
    {\usebibmacro{date:print}}
    {}%
  \printfield{addendum}%
  \newunit\newblock
  \printfield{pubstate}}

  \renewbibmacro*{volume+number+eid}{%
  \printfield{volume}%
  \setunit*{\addnbthinspace}
  \printfield{number}%
  \setunit{\addcomma\space}%
  \printfield{eid}}
\DeclareFieldFormat[article]{number}{\mkbibparens{#1}}

\AtEveryBibitem{
  \ifentrytype{article}{
    \clearlist{language}
    \clearfield{eprint}
    \clearfield{isbn}
    \clearfield{issn}
    \clearfield{issue}
    \clearfield{month}
    \clearfield{url}
    \clearfield{doi}
  }{}
  \ifentrytype{book}{
    \clearlist{address}
    \clearlist{location}
    \clearlist{language}
    \clearfield{url}
  }{}
  \ifentrytype{inbook}{
    \clearfield{isbn}
    \clearfield{url}
  }{}
  \ifentrytype{inproceedings}{
    \clearfield{isbn}
    \clearfield{issn}
	  \clearfield{doi}
    \clearlist{publisher}
    \clearname{editor}
    \clearlist{location}
    \clearfield{url}
  }{}
  \ifentrytype{misc}{
    \clearfield{url}
  }{}
  \ifentrytype{unpublished}{
    \clearfield{isbn}
    \clearfield{issn}
	  \clearfield{doi}
    \clearlist{publisher}
    \clearname{editor}
    \clearlist{address}
    \clearname{primaryclass}
  }{}
}

%% file: Preliminaries.tex
\section{Preliminaries}\label{sec:prelim}
In this section we introduce notations and definitions, and present some useful lemmas.

\subsection{Notations}
\paragraph{General notations.}
For any integer $N$ we write $[N]=\set{1}{N}$. Define $\Real[\poly(n)]$ as the set of all real numbers with binary expansion of polynomial length. More precisely,
for any function $f\colon\Nat\to\Nat$, we define the set 
\[
	\Real[f(n)]=\left\{\pm\left( a_0+\sum_{i=1}^{f(n)} a_i2^{i}+\sum_{i=1}^{f(n)} b_i2^{-i}\right)\:\:\Big\vert\:\: a_{0},\ldots,a_{f(n)},b_1,\ldots,b_{f(n)}\in \{0,1\}\right\}\subseteq \Real
\]
of 
all real numbers with binary expansion of polynomial length $2f(n)+2$ (including one bit for encoding the sign). Then $\Real[\poly(n)]$ is the union of the $\Real[f(n)]$ for all polynomial functions $f(n)$.  We define $\Comp[f(n)]$ and $\Comp[\poly(n)]$ similarly, by requiring that both the real part and the imaginary part are in $\Real[f(n)]$ and $\Real[\poly(n)]$, respectively.\vspace{-2mm}

\paragraph{Vectors and matrices.}
In this paper we consider vectors in $\Comp^{N}$ and matrices in $\Comp^{N\times N}$, for some integer $N$, and write $n=\ceil{\log_2(N)}$. Note that this notation is consistent with the notation of Section \ref{sec:intro}, where we considered the special case $N=2^n$.
We usually write quantum states (i.e., unit-norm vectors) using Greek letters and Dirac notation, e.g., we use $\ket{\psi}$ or $\ket{\varphi}$. We write arbitrary vectors (i.e., vectors of arbitrary norm) using Roman letters, e.g., we use $v$ or $w$. 

For a matrix $A\in\Comp^{N\times N}$ and any $\ell\in[N]$, we denote the $\ell$-th row of $A$ by $A[\ell,\cdot]$. The matrix $A$ is normal if it can be written $A=UDU^{-1}$ where $D$ is a diagonal matrix with real entries and $U$ is a unitary matrix. 
We use $\norm{A}$ to denote the spectral norm of $A$, which is defined for a normal matrix as the maximum magnitude of the eigenvalues of $A$ (and defined as square root of the maximum eigenvalue of $A^\ast A$, where $A^\ast$ denotes the conjugate transpose of $A$, in general). This norm is submultiplicative, i.e., the inequality $\norm{AB}\le \norm{A}\norm{B}$ holds for any matrices $A,B\in\Comp^{N\times N}$. We also have $\abs{u^\ast A v}\le \norm{A}\norm{u}\norm{v}$ for any $u,v\in\Comp^N$, where $\norm{u}$ and $\norm{v}$ denote the Euclidean norms of $u$ and $v$, respectively.\vspace{-2mm}

\paragraph{Eigenvalues and overlap.}
Consider a normal matrix $A\in\Comp^{N\times N}$. Let 
\[
	A=\sum_{i=1}^{2^n}\lambda_i\ket{u_i}\bra{u_i}
\]
be its eigenvalue decomposition,
with eigenvalues
$\lambda_1\le \lambda_2\le\cdots\le\lambda_{2^n}$ and corresponding orthonormal eigenvectors $\ket{u_1},\ldots,\ket{u_{2^n}}$. We denote by $\en{A}=\lambda_1$ the smallest eigenvalue of $A$.
For any $\sigma\ge 0$, let us write $S(A,\sigma)=\{i\in[N]\:|\:\lambda_i(A)\le \en{A}+\sigma\}$. For any vector $w\in\Comp^N$,
let 
\[
	\overlapf{A,w}{\sigma}=\sqrt{\sum_{i\in S(A,\:\sigma)}\abs{\braket{u_i}{w}}^2}
\]
denote the overlap of $w$ with the eigenspace corresponding to eigenvalues in $[\en{A},\en{A}+\sigma]$. Note that the standard definition of the overlap (used in Section \ref{sec:intro}) corresponds to the case $\sigma=0$. 

\subsection{Access to vectors and matrices}
We now define the notions of access to vectors and matrices needed for this work. These notions are similar to prior works on dequantization \cite{Chia+JACM22,LeGall23,TangSTOC19}, but we need to precisely discuss the encoding length and the space complexity.

We start with query access to a vector.

\begin{definition}\label{def:q-access} 
	We have query access to a vector $w\in \Comp^{N}$ with encoding length $\len{w}$ and costs $\qt{w}$ and $\qs{w}$ if 
	\begin{itemize}
		\item[1.] for each $i\in[N]$, we have $w_i\in\Comp[\len{w}]$;
		\item[2.] for any $i\in[N]$, the coordinate $w_i$ can be obtained in $\qt{w}$ time and $\qs{w}$ space.
	\end{itemize}
	If $\len{w},\qt{w},\qs{w}\le \poly(n)$,
	we simply say that we have query access to $w$.
\end{definition}

Next, we introduce the stronger notion of sample-and-query access to a vector.
\begin{definition}\label{def:samplable}
We have sample-and-query access to a vector $w\in \Comp^{N}$ if 
\begin{itemize}
	\item[1.] we have query access to $w$;
	\item[2.] we can compute in $\poly(n)$ time\footnote{Since a polynomial upper bound on the time complexity implies a polynomial upper bound on the space complexity, hereafter we omit to explicitly mention that the space complexity is $\poly(n)$ as well.} a sample from the distribution $p\colon [N]\to [0,1]$ such that \[p(i)=\frac{\abs{w_i}^2}{\norm{w}^2}\] for each $i\in[N]$.
\end{itemize}
\end{definition}
\noindent
When $\norm{w}=1$,
Item 2 in \cref{def:samplable} states that we can efficiently sample from the same distribution as the distribution obtained when measuring the quantum state $\sum_{i=1}^Nw_i\ket{i}$ in the computational basis.  

We extend the notion of query access to matrices as follows:
\begin{definition}\label{def:ma}
	We have query access to a matrix $B\in \Comp^{N\times N}$ if
	\begin{itemize}
		\item[1.]for each $(i,j)\in [N]\times[N]$, we have 
		$B[i,j]\in\Comp[\poly(n)]$;
		\item[2.] for any  $(i,j)\in [N]\times[N]$, the entry $B[i,j]$ can be obtained in $\poly(n)$ time; 
		\item[3.] for any $i\in [N]$, the number $s_i$ of nonzero entries in $B[i,\cdot]$ can be obtained in $\poly(n)$ time; 
		\item[4.] for any $i\in [N]$ and any $\ell\in[s_i]$, the $\ell$-th nonzero entry of $B[i,\cdot]$ can be obtained in $\poly(n)$ time.
	\end{itemize}
\end{definition}
Items 3 and 4 in \cref{def:ma} are needed to deal with sparse matrices.

\subsection{Local Hamiltonians and matrix decompositions}
We give below technical details about the description of local Hamiltonians and the matrix decompositions introduced in this work.\vspace{-2mm}

\paragraph{Description of local Hamiltonians.}
A $k$-local Hamiltonian acting on $n$ qubits is a Hermitian matrix $H\in\Comp^{N\times N}$ with $N=2^n$ that can be written as  
\[
H=
\sum_{i=1}^m H_i
\]
with $m=\poly(n)$,
where each term $H_i$ is an Hermitian matrix acting non-trivially on at most~$k$ qubits. Each $H_i$ can be described by a $2^k\times 2^k$ matrix representing its action on the $k$ qubits on which it acts non-trivially. We assume that each entry of this description is in $\Comp[\poly(n)]$. 
This description is given as input. For convenience we also assume that we know $\norm{H_i}$ for each $i\in[N]$.\footnote{Note that each $\norm{H_i}$ can be computed from its description as a $2^k\times 2^k$ matrix. The computation is efficient when $k$ is small, e.g., for $k=O(\log n)$, which is the most interesting regime for the local Hamiltonian problem.}

\paragraph{Matrix decomposition.}
Here is the complete definition of the matrix decomposition we consider.

\begin{definition}\label{def:decomp}
	For a matrix $A$, an integer $s\ge 0$ and a real number $\upw\in\Real[\poly(n)]$, an $(s,\upw)$-decomposition of $A$ is a decomposition 
	\[
		A=\sum_{i=1}^m A_i  
		\hspace{3mm}\textrm{with}\hspace{2mm}
		\sum_{i=1}^m \norm{A_i}\le \upw
	\] 
	in which $A_i$ is an $s$-sparse matrix for each $i\in[m]$. We always (implicitly) assume the following:
	\begin{itemize}
		\item for each $i\in [m]$, we have query access to the matrix $A_i$;
		\item we know bounds $\upw_1,\ldots,\upw_m\in\Real[\poly(n)]$ such that $\norm{A_i}\le \upw_i$ for each $i\in[m]$ and $\sum_{i=1}^m \upw_i=\upw$.
	\end{itemize}
	  If $\upw=1$, we simply call the decomposition an $s$-decomposition.
\end{definition}

\subsection{Lemmas}
We present four lemmas that are needed to prove our results.
 
The first lemma is the ``powering lemma'' from \cite{Jerrum1986} to amplify the success probability of probabilistic estimators (the formulation below for complex numbers is from \cite[Lemma 3]{LeGall23}):
\begin{lemma}[Powering lemma]\label{lemma:powering}
Consider a randomized algorithm that produces an estimate $\tilde \mu$ of a complex-valued quantity $\mu$ such that $\abs{\tilde \mu-\mu}\le \varepsilon$ holds with probability at least $3/4$. Then, for any $\delta>0$, it suffices to repeat $O(\log(1/\delta))$ times the algorithm and take both the median of the real parts and the median of the imaginary parts to obtain an estimate $\hat \mu$ such that $\abs{\hat \mu-\mu}\le \sqrt{2}\varepsilon$ holds with probability at least $1-\delta$.
\end{lemma}

To perform eigenvalue estimation we will need a low-degree polynomial that approximates well the ``rectangle'' function. We will use the following result from \cite{Gilyen+STOC19}.\footnote{\cref{lemma:poly-cons} follows by taking 
$t=\tau+\theta/2$ and $\delta'=\theta/2$ in Lemma 29 of \cite{Gilyen+STOC19}.
The computability of the polynomial is discussed explicitly in \cite[Appendix A.3]{LeGall+23}.} 
\begin{lemma}[Lemma 29 in \cite{Gilyen+STOC19}]\label{lemma:poly-cons}
    For any $\xi\in(0,1]$, any $\tau\in[0,1)$ and any $\theta\in(0,1-\tau]$, there exists an efficiently computable polynomial $P\in\Real[x]$ of degree $O\!\left(\frac{1}{\theta}\log(1/\xi)\right)$ such that $|P(x)|\in [0,1]$ for all  $x\in[-1,1]$ and
    \begin{equation}\label{eq:P}
        \begin{cases}
        P(x)\in [1-\xi,1]& \textrm{ if } x\in[0,\tau],\\
        P(x)\in [0,\xi]& \textrm{ if } x\in[\tau+\theta,1].\\
        \end{cases}
    \end{equation}
\end{lemma}

We will use the following result from \cite{Sherstov2013} that gives an upper bound on the coefficients of polynomials bounded in the interval $[-1,1]$ (such as the polynomial from \cref{lemma:poly-cons}). 
\begin{lemma}[Lemma 4.1 in \cite{Sherstov2013}]\label{lemma:poly}
	Let $P(x)=\sum_{i=0}^d a_i x^i$ be a univariate polynomial of degree $d$ such that $\abs{P(x)}\le1$ for all $x\in[-1,1]$. Then
	\[
		\sum_{i=0}^d \abs{a_i}\le 4^d.
	\]
\end{lemma}

Finally, we discuss how to classically estimate the inner product $\braket{\psi}{w}$ given sample-and-query access to a quantum state $\ket{\psi}$ and query access to a vector $w$. More precisely, we are considering the following problem:

\begin{center}
	\ovalbox{
	\begin{minipage}{13 cm} \vspace{2mm}
	
	\noindent\hspace{3mm}$\IP(\varepsilon,\delta)$\hspace{5mm}{\tt(Estimation of Inner Product)}\\\vspace{-3mm}
	
	\noindent\hspace{3mm} Input: $\bul$ sample-and-query access to a quantum state $\ket{\psi}\in \Comp^{N}$

	\noindent\hspace{15mm}
	$\bul$ query access to a vector $w\in \Comp^{N}$ with encoding length $\len{w}$ and
	
	\hspace{91mm}
	costs $\qt{w}$ and $\qs{w}$

	\vspace{2mm}
	
	\noindent\hspace{3mm} Output: an estimate $a\in\Comp$ such that 		
		\[
			\abs{a-\bra{\psi}w\rangle}\le \varepsilon\,\norm{w}
		\]

	\noindent\hspace{25mm}
	holds with probability at least $1-\delta$
	\vspace{2mm}
	\end{minipage}
	}
\end{center}
\vspace{2mm}

\noindent Prior works on dequantization \cite{Chia+JACM22,LeGall23,TangSTOC19} have shown how to solve this problem efficiently. It can be easily checked that these approaches are space-efficient as well, leading to the following statement.  For completeness we give a proof in Appendix \ref{app}. 
\begin{lemma}\label{lemma:Tang}
	For any $\varepsilon\in(0,1]$ and any $\delta\in(0,1]$, 
	the problem $\IP(\varepsilon,\delta)$ can be solved classically in time
	\[
		\myO{\qt{w}\,\varepsilon^{-2}\log (1/\delta)}
	\] 
	and space
	\[
		\qs{w}+\myO{\Big(\len{w}+\log(1/\varepsilon)\Big)\log(1/\delta)}.
	\]
\end{lemma}

%% file: IMM.tex
\section{Iterated Matrix Multiplication}\label{sec:IMM}
In this section we show how to classically estimate the inner product $\bra{\psi}B_r\cdots B_1\ket{\varphi}$ for sparse matrices $B_1,\ldots,B_r$ and two quantum states $\ket{\psi}$ and $\ket{\varphi}$ to which we have classical access. 
More precisely, we consider the following problem.

\begin{center}
	\ovalbox{
	\begin{minipage}{14.8 cm} \vspace{2mm}
	
	\noindent\hspace{3mm}$\MM(s,r,\varepsilon,\delta)$\hspace{5mm}{\tt(Estimation of Iterated Matrix Multiplication)}\\\vspace{-3mm}
	
	\noindent\hspace{3mm} Input: $\bul$ query access to $s$-sparse matrices $B_1,\ldots,B_r\in\Comp^{N\times N}$ 
	

	\noindent\hspace{15mm}
	$\bul$ query access to a quantum state $\ket{\varphi}\in \Comp^{N}$
	
	\noindent\hspace{15mm}
	$\bul$ sample-and-query access to a quantum state $\ket{\psi}\in \Comp^{N}$
	\vspace{2mm}
	
	\noindent\hspace{3mm} Output: an estimate $\hat E\in\Comp$ such that 		
		\[
			\abs{\hat E - \bra{\psi}B_r\cdots B_1\ket{\varphi}}\le \varepsilon\, \norm{B_1}\cdots \norm{B_r}
		\]

	\noindent\hspace{25mm}
	holds with probability at least $1-\delta$
	\vspace{2mm}
	\end{minipage}
	}
\end{center}
\vspace{2mm}

Here is the main result of this section.

\begin{proposition}\label{th:IMM}
	For any $s\ge 1$, any $r\ge 1$, any $\varepsilon\in(0,1]$ and any $\delta\in(0,1]$, the problem $\MM(s,r,\varepsilon,\delta)$ can be solved classically in time
	\[\myO{  s^r\varepsilon^{-2}\log(1/\delta)}\] and space
	\[
		\myO{r^2+\big(r+\log(1/\varepsilon)\big)\log(1/\delta)}.
	\]
\end{proposition}

The proof of \cref{th:IMM} is based on the following lemma, which can be seen as a space-efficient version of the approach for iterated matrix multiplication used in \cite{Gharibian+SICOMP2023,Schwarz2013}.

\begin{lemma}\label{lemma2}
	There is a classical algorithm that implements query access to the vector $B_r\cdots B_1\ket{\varphi}$ with encoding length $\len{B_r\cdots B_1 \ket{\varphi}}=\myO{r}$ and costs
	$\qt{B_r\cdots B_1\ket{\varphi}}=\myO{s^{r}}$ and $\qs{B_r\cdots B_1\ket{\varphi}}=\myO{r^2}$.
\end{lemma}
\begin{proof}
	Here is the main idea: to obtain the $\ell$-th entry of $B_r\cdots B_1\ket{\varphi}$, we only need to know the~$s$ nonzero entries of the $\ell$-th row of $B_r$, which can be queried directly, together with the corresponding entries in the vector $B_{r-1}\cdots B_{1}\ket{\varphi}$, which can be computed recursively. The algorithm is described in pseudocode below. 
	
	\begin{center}
		\begin{minipage}{14 cm} \vspace{2mm}
		\begin{algorithm}[H]
			\SetAlgorithmName{}{}{List of Algorithms}
			\nonl \hspace{-4mm}Algorithm $\Aa(\ell,r)$ 
			\hspace{4mm}\tcp*[l]{computes $\bra{\ell}B_r\cdots B_1 \ket{\varphi}$}
	
			\uIf{$r=0$}
				{{\bf return} $\braket{\ell}{\varphi}$\,;\label{step:phi}}
			\Else
				{
				$z\gets 0$\,;
	
				get the number of nonzero entries of the row $B_r[\ell,\cdot]$ and write it $s'$\,;
	
				\For{$t$ {\bf from }$1$ \KwTo $s'$\label{step:forbegin}}
					{
					get the index of the $t$-th nonzero entry of $B_r[\ell,\cdot]$ and write it $j$\,;
	
					$x\gets B_r[\ell,j]$\,;\label{step:x}
					\hspace{11mm}\tcp*[l]{queries $\bra{\ell}B_{r}\ket{j}$}
	
					$y\gets \Aa\left(j,r-1\right)$\,;
					\hspace{4mm}\tcp*[l]{computes recursively $\bra{j}B_{r-1}\cdots B_1\ket{\varphi}$}
					\label{step:y}
	
					$z\gets z+ x \cdot y$\,;\label{step:z}
					}
				{\bf return} $z$\,;
				}
			\end{algorithm}
		\end{minipage}
	\end{center}
	
	We first analyze the correctness of the algorithm. Let $j_1,\ldots,j_{s'}$ represent the indices of the nonzero entries of the row $B_r[\ell,\cdot]$. Since
	\[
		\bra{\ell}B_r\cdots B_1 \ket{\varphi} =\sum_{t=1}^{s'} \bra{\ell}B_r\ket{j_t} \bra{j_t}B_{r-1}\cdots B_1\ket{\varphi},
	\]
	Algorithm $\Aa(\ell,r)$ outputs $\bra{\ell} B_r\cdots B_1 \ket{\varphi}$. 
	
	Let $\Tt(r)$ denote the running time of this procedure. 
	We have
	\[
	\Tt(r)\le s\Tt(r-1)+\myO{s},
	\]
	and thus $\Tt(r)=\myO{s^r}$. 
	For each $j\in[N]$, the entry $\braket{j}{\psi}$ has a $\poly(n)$-bit binary expansion. Each entry of the matrices $B_1,\ldots,B_r$ also has a $\poly(n)$-bit binary expansion. This implies that $\len{B_r\cdots B_1 \ket{\varphi}}=\myO{r}$.
	We finally consider the space complexity. The recursion tree has depth~$r$. At each level of the recursion, 
	the values $x$, $y$ and $z$ at Steps~\ref{step:x}, \ref{step:y} and \ref{step:z} can be stored in $\myO{r}$ bits, and we need one~$O(\log s)$-bit counter for storing the current value of~$t$. The overall space complexity of the algorithm is thus $\myO{r(r+\log s)}=\myO{r^2}$.
	\end{proof}

\cref{th:IMM} is obtained by applying \cref{lemma:Tang} to the vector $w=B_r\cdots B_1\ket{\varphi}$, for which we can implement query access from \cref{lemma2}:
\begin{proof}[Proof of \cref{th:IMM}]
	From \cref{lemma2} we have query access to the vector $w=B_r\cdots B_1\ket{\varphi}$ with encoding length $\len{w}=\myO{r}$ and costs $\qt{w}=\myO{s^{r}}$ and $\qs{w}=\myO{r^2}$. 
	Using \cref{lemma:Tang}, we can then compute an estimate $a\in\Comp$ such that
	\[
		\abs{a-\bra{\psi}B_r\cdots B_1\ket{\varphi}}
		\le \varepsilon\, \norm{B_r\cdots B_1\ket{\varphi}}
		\le \varepsilon\, \norm{B_1}\cdots \norm{B_r}
	\]
	holds with probability at least $1-\delta$.
	The time and space complexity are 
	\[\myO{s^r\varepsilon^{-2}\log(1/\delta)}\] and \[\myO{r^2+\big(r+\log(1/\varepsilon)\big)\log(1/\delta)},\] respectively.
\end{proof}

%% file: PT.tex
\section{Polynomial Transformations of Decomposable Matrices}\label{sec:PT}
In this section we show how to classically estimate the inner product $\bra{\psi}P(A)\ket{\varphi}$ for a matrix $A$ with an $s$-decomposition, a polynomial $P$, and two quantum states $\ket{\psi}$ and $\ket{\varphi}$ to which we have classical access. More precisely, we consider the following problem.

\begin{center}
	\ovalbox{
	\begin{minipage}{13.0 cm} \vspace{2mm}
	
	\noindent\hspace{3mm}$\EST(s,d,\eta)$ \hspace{5mm}{\tt(Estimation of Polynomial Transformation)}\\\vspace{-3mm}
	
	\noindent\hspace{3mm} Input: $\bul$ a matrix $A\in\Comp^{N\times N}$ with an $s$-decomposition  

	\noindent\hspace{15mm}
	$\bul$ a polynomial $P\in\Real[x]$ of degree $d$ with $|P(x)|\le 1$ $~\forall x\in[-1,1]$

	\noindent\hspace{15mm}
	$\bul$ query access to quantum state $\ket{\varphi}\in \Comp^{N}$
	
	\noindent\hspace{15mm}
	$\bul$ sample-and-query access to a quantum state $\ket{\psi}\in \Comp^{N}$
	\vspace{2mm}
	
	\noindent\hspace{3mm} Output: an estimate $\hat{E}\in \Comp$ such that 
	\begin{equation}\label{eq:PTcond}
		\abs{\hat E-\bra{\psi} P(A)\ket{\varphi}}\le \eta
	\end{equation}

	\noindent\hspace{25mm}
	holds with probability at least $1-1/\exp(n)$
	\vspace{2mm}
	\end{minipage}
	}
	\end{center}
	\vspace{2mm}

	Here is the main result of this section.
\begin{proposition}\label{th:PT}
	For any $s\ge 2$, any $d\ge 1$ and any $\eta\in(0,1]\cap\Real[\poly(n)]$, the problem $\EST(s,d,\eta)$ can be solved classically in time
	\[
		\myO{s^{c\cdot d} \eta^{-4}}
	\] 
	time, for some universal constant $c>0$, and space $\myO{d^2}$\,.
\end{proposition}

The proof of \cref{th:PT} is based on the following lemma, whose proof is given after the proof of the proposition.

\begin{lemma}\label{th:t1}
	For any $r\in\set{0}{d}$, any $\eta\in(0,1]\cap\Real[\poly(n)]$ and any $\delta\in(0,1]$, there is a classical algorithm that computes an estimate $\hat E_r\in \Comp$ such that 
	\[
		\abs{\hat E_r - \bra{\psi}A^{r}\ket{\varphi}}\le\frac{\eta}{4^d} 
	\] 		
	holds with probability at least $1-\delta$ in time
	\[
		\myO{s^{r} 2^{8d}\eta^{-4}d\log(1/\delta)}
	\]
	and space
	\[
		\myO{d^2+d\log(1/\delta)}.
	\]
\end{lemma}

\begin{proof}[Proof of \cref{th:PT}]
	Let us write the polynomial $P$ as \[P(x)=\sum_{r=0}^d a_r x^d.\] 
	For any $\delta'\in(0,1]$, we describe how to compute an estimate $\hat E$ such that \cref{eq:PTcond} holds with probability at least $1-\delta'$. Taking $\delta'=1/\exp(n)$ then proves the proposition. 

	For each $r\in\{0,\ldots,d\}$ such that $a_r\neq 0$, we apply \cref{th:t1} with $\delta=\frac{\delta'}{d+1}$ to obtain an approximation $\hat E_r$ of $\bra{\psi}A^{r}\ket{\varphi}$ such that 
	\[
		\Pr\left[\abs{\hat E_r - \bra{\psi}A^{r}\ket{\varphi}}\le \frac{\eta}{4^d}\right]\ge 1-\frac{\delta'}{d+1}.
	\]
	For each $r\in\{0,\ldots,d\}$ such that $a_r= 0$, we set $\hat E_r=0$.
	We then output \[\hat E=\sum_{r=0}^d a_r\hat E_r.\] From the union bound and the triangle inequality, with probability at least $1-\delta'$ we have
	\[
		\abs{\hat E - \bra{\psi}P(A)\ket{\varphi}}\le
		\sum_{r=0}^d\abs{a_r}\abs{E_r - \bra{\psi}A^{r}\ket{\varphi}}\le \eta,
	\]
	where we used \cref{lemma:poly} to derive the last inequality.

	The time complexity is
	\begin{align*}
		\myO{\left(\sum_{r=0}^d s^{r}\right)2^{8d} \eta^{-4}d\log\left(\frac{d}{\delta'}\right)}
		&=\myO{s^d2^{8d} \eta^{-4}d\log\left(\frac{d}{\delta'}\right)}\\
		&=\myO{s^{c\cdot d} \eta^{-4}},
	\end{align*}
	for some universal constant $c>0$.
	The space complexity is $\myO{d^2+d\log(d/\delta')}=\myO{d^2}$.
\end{proof}

\begin{proof}[Proof of \cref{th:t1}]
	As in \cref{def:decomp}, we write the $s$-decomposition of $A$ as
	\[
		A=\sum_{i=1}^m A_i\,,  
	\] 
	where each $A_i$ is an $s$-sparse matrix such that $\norm{A_i}\le\upw_i$, with 
	\begin{equation}\label{eq:norm}
	\upw_1+\cdots+\upw_m = 1.
	\end{equation}
	Consider the probability distribution $p\colon[m]^r\to[0,1]$ defined 
	as \[p(x)=\upw_{x_1}\cdots\upw_{x_r}\] for any $x\in[m]^r$ (\cref{eq:norm} guarantees that this is a probability distribution). 
	Define a random variable $X$ as follows: sample a vector $x$ from the distribution $p$, 
	and set 
	\[
	X=\frac{\bra{\psi}A_{x_1}\cdots A_{x_r}\ket{\varphi}}{p(x)}.
	\]
	Repeat the above procedure $t=\ceil{64\cdot4^{2^d}/\eta^2}$ times and output the mean. Let $Y$ denote the corresponding complex random variable. 
	We have 
	\begin{align*}
		\ex{Y}&=\ex{X}=
		\sum_{x\in[m]^r}\bra{\psi}A_{x_1}\cdots A_{x_r}\ket{\varphi}
		= 
		\bra{\psi}A^r\ket{\varphi}
	\end{align*}
	and
	\begin{align*}
		\var{Y}
		&\le
		\frac{1}{t}
		\ex{\abs{X}^2}\\
		&=
		\frac{1}{t}\sum_{x\in[m]^r}\frac{\abs{\bra{\psi}A_{x_1}\cdots A_{x_r}\ket{\varphi}}^2}{\upw_{x_1}\cdots\upw_{x_r}}\\
		&\le
		\frac{1}{t}\sum_{x\in[m]^r}\frac{\norm{A_{x_1}\cdots A_{x_r}}^2}{\upw_{x_1}\cdots\upw_{x_r}}\\
		&\le
		\frac{1}{t}\sum_{x\in[m]^r}\upw_{x_1}\cdots\upw_{x_r}\\
		&=
		\frac{1}{t}\left(\upw_1+\cdots+\upw_m\right)^r\\
		&=
		\frac{1}{t}\,.
	\end{align*}
	From Chebyshev's inequality, we thus obtain:
	\begin{equation}\label{eq:ub1}
	\pr{\abs{Y-\bra{\psi}A^r\ket{\varphi}}\ge \frac{\eta}{2\sqrt{2}\cdot 4^d}}\le\frac{8\cdot 4^{2d}}{\eta^2 t}\le\frac{1}{8}.
	\end{equation}

	We cannot directly use this strategy since we do not know $\bra{\psi}A_{x_1}\cdots A_{x_r}\ket{\varphi}$. Instead, we estimate this quantity using \cref{th:IMM}. This leads to the following algorithm.
	\begin{center}
		\begin{minipage}{13.3 cm} \vspace{2mm}
		\begin{algorithm}[H]
			\SetAlgorithmName{}{}{List of Algorithms}
			\nonl \hspace{-4mm}Algorithm $\Bb(\eta)$ 
			\hspace{4mm}\tcp*[l]{estimates $\bra{\psi}A^r \ket{\varphi}$ with precision $\frac{\eta}{\sqrt{2}\cdot4^d}$}
			
			$t\gets\ceil{64\cdot 4^{2d}/\eta^2}$;

			$z\gets 0$\,;
			
			\For{$i$ {\bf from }$1$ \KwTo $t$}
					{
					Take a vector $x$ according to the distribution $p$.\label{step:choice}

					Use \cref{th:IMM} for the problem $\MM\big(s,r,\frac{\eta}{2\sqrt{2}\cdot 4^d},\frac{1}{8t}\big)$ to
					compute an estimate $\alpha\in \Comp$ of $\bra{\psi}A_{x_1}\cdots A_{x_r}\ket{\varphi}$\,;\label{step:alpha}

					$z\gets z+\frac{\alpha}{t\cdot p(x)}$\,;
					}
				{\bf return} $z$\,;\label{step:output}
				
			\end{algorithm}
		\end{minipage}
	\end{center}

	The complexity of Algorithm $\Bb(\eta)$ is dominated by the computation at Step \ref{step:alpha}, which is done $t$ times. 
	From \cref{th:IMM}, we obtain the upper bounds
	\[
		\myO{t\cdot s^{r} 2^{4d}\eta^{-2}\log(8t)}
		=
		\myO{s^{r} 2^{8d}\eta^{-4}d}
	\]
	and 
	\[
		\myO{r^2+\left(r+\log\left(\frac{2\sqrt{2}\cdot 4^d}{\eta}\right)\right)\log(8t)}=\myO{r^2+d^2}=\myO{d^2}
	\]
	on the time and space complexities, respectively.

	We now analyze the correctness of Algorithm $\Bb(\eta)$.
	Let $Z$ be the random variable corresponding to the output of Step \ref{step:output} when at Step \ref{step:choice} the vectors $x$'s are the same vectors as in the random variable $Y$.
		For any choice of $x$ at Step \ref{step:choice}, the estimate $\alpha$ of Step \ref{step:alpha} satisfies
		\begin{equation}\label{eq:beta}
		\abs{\alpha-\bra{\psi}A_{x_1}\cdots A_{x_r}\ket{\varphi}}\le \frac{\eta\,\upw_{x_1}\cdots\upw_{x_r}}{2\sqrt{2}\cdot 4^d}
	\end{equation}
	with probability at least $1-1/(8t)$.
		Under the condition that Inequality (\ref{eq:beta}) is always satisfied during the $t$ repetitions, we have  
	\[
		|Z-Y|\le\sum_{x}\frac{1}{t\cdot p(x)}\frac{\eta\,\upw_{x_1}\cdots\upw_{x_r}}{2\sqrt{2}\cdot 4^d}
		=
		\frac{\eta}{2\sqrt{2}\cdot 4^d},
	\]
	where the sum is over the $t$ vectors $x$ chosen at Step \ref{step:choice}.
	By the union bound, we thus have 
	\begin{equation}\label{eq:ub2}
		\Pr\left[|Z-Y|>\frac{\eta}{2\sqrt{2}\cdot 4^d}\right]\le \frac{1}{8}.
	\end{equation}

	Combining \cref{eq:ub1} and \cref{eq:ub2} gives
	\begin{align*}
		\Pr\left[\abs{Z-\bra{\psi}A^r\ket{\varphi}}\le \frac{\eta}{\sqrt{2}\cdot 4^d}\right]
		&\ge
		\Pr\left[|Z-Y|\le\frac{\eta}{2\sqrt{2}\cdot 4^d}\:\:\textrm{ and }\:\abs{Y-\bra{\psi}A^r\ket{\varphi}}\le \frac{\eta}{2\sqrt{2}\cdot 4^d}\right]\\
		&\ge 
		1-
		\Pr\left[|Z-Y|>\frac{\eta}{2\sqrt{2}\cdot 4^d}\right]
		- 
		\Pr\left[\abs{Y-\bra{\psi}A^r\ket{\varphi}}> \frac{\eta}{2\sqrt{2}\cdot 4^d}\right]\\
		&\ge  \frac{3}{4}.
	\end{align*}
	We can then use Lemma \ref{lemma:powering} to obtain an estimate $\hat E_r\in \Comp$ such that 
	\[	
		\abs{\hat E_r - \bra{\psi}A^{r}\ket{\varphi}}\le\frac{\eta}{4^d} 
	\] 		
	holds with probability at least $1-\delta$. Lemma \ref{lemma:powering} introduces a $\log(1/\delta)$ factor in the time complexity and an additive $\myO{d\log(1/\delta)}$ term in the space complexity since for the computation of the medians we need to store $O(\log(1/\delta))$ values, each requiring $\myO{d}$
	bits.\footnote{Here we are using the assumption $\kappa_i\in\Comp[\poly(n)]$ for all $i\in[m]$, see \cref{def:decomp}. This implies that $p(x)$ can be encoded in $\myO{r}$ bits and the output of Procedure $\Bb(\eta)$ in $\myO{d}$ bits.} 
\end{proof}

%% file: EV.tex
\section{Eigenvalue Estimation}\label{sec:EV}
In this section we use the results proved in Section \ref{sec:PT} to estimate the smallest eigenvalue of a normal matrix. 
We describe the most general problem we are solving in Section \ref{sub:EV1} and then prove Theorems \ref{th:main-informal} and \ref{cor:main-informal} in Section~\ref{sub:EV2}.

\subsection{General result}\label{sub:EV1}
We consider the problem of estimating the smallest eigenvalue of a normal matrix with an $(s,\upw)$-decomposition, given classical access to a guiding state. Here is the formal description\footnote{
	The choice of $\frac{\varepsilon}{2}\upw$ in the overlap is somehow arbitrary: we could have chosen $\sigma\upw$ for any $\sigma\in[0,\varepsilon)$ instead.  
	The exponent in the complexity of \cref{th:EV} would then become 
	$O\left(\frac{\log(1/\chi)}{\varepsilon-\sigma}\right)$ instead of $O\left(\frac{\log(1/\chi)}{\varepsilon}\right)$. 
} of the problem:

\begin{center}
	\ovalbox{
	\begin{minipage}{15.2 cm} \vspace{2mm}
	
	\noindent\hspace{3mm}$\EV(s,\chi,\varepsilon)$ \hspace{5mm}{\tt(Estimation of the Smallest Eigenvalue)}\\\vspace{-3mm}
	
	\noindent\hspace{3mm} Input: $\bul$ a normal matrix $A\in\Comp^{N\times N}$ with an $(s,\upw)$-decomposition (for any $\kappa)$ 
	\vspace{1mm}

	\noindent\hspace{15mm}
	$\bul$ sample-and-query access to a quantum state $\ket{\psi}\in \Comp^{N}$ with $\overlapf{A,\ket{\psi}}{\frac{\varepsilon}{2}\kappa}\ge\chi$

	\vspace{3mm}
	
	\noindent\hspace{3mm} Output: an estimate $E^\ast\in \Real$ such that 
	\begin{equation}\label{eq:EVcond}
		\abs{E^\ast-\en{A}}\le \varepsilon \, \upw
	\end{equation}

	\noindent\hspace{25mm}
	holds with probability at least $1-1/\exp(n)$
	\vspace{2mm}
	\end{minipage}
	}
\end{center}

We prove the following theorem:
\begin{theorem}\label{th:EV}
	For any $s\ge 2$ and any $\varepsilon,\chi\in(0,1]\cap\Real[\poly(n)]$, 
	the problem $\EV(s,\chi,\varepsilon)$ can be solved classically in time
	\[
		\myO{s^{\frac{c'\log(1/\chi)}{\varepsilon}}}
	\]
	time, for some universal constant $c'>0$, and space $\myO{1/\varepsilon^2}$.
\end{theorem}

\begin{proof}[Proof of \cref{th:EV}]
	Let us define $A'=\frac{1}{2}(I+\frac{A}{\upw})$ and write $T=\ceil{4/\varepsilon}$. The $(s,\upw)$-decomposition of $A$ gives an $(s+1)$-decomposition of $A'$.
	Observe that $A'$ has eigenvalues in the interval $[0,1]$. 
	The main idea is to divide this interval into $T$ subintervals of length at most $\varepsilon/4$ and find in which subinterval $\en{A'}$ lies in. Since $\en{A}=2\upw\en{A'}-\upw$, this will give an estimate of $\en{A}$.

	Concretely, for any $t\in\set{0}{T-1}$, we consider the following test that checks if $\en{A'}$ is ``approximately'' smaller than  $t\frac{\varepsilon}{4}$. The approximation comes from the use of an estimator at Step~\ref{step:est} --- details of the implementation of this step are discussed later. 
	\begin{center}
		\begin{minipage}{13.5 cm} \vspace{2mm}
		\begin{algorithm}[H]
			\setstretch{1.35}
			\SetAlgorithmName{}{}{List of Algorithms}
			\nonl \hspace{-4mm}$\Test(t)$ 
			\hspace{2mm}\tcp*[l]{checks if $\en{A'}$ is (approximately) smaller than  $t\frac{\varepsilon}{4}$}

					Let $P$ be the polynomial of \cref{lemma:poly-cons} with 
					$\tau=t\frac{\varepsilon}{4}$,
					$\theta=\frac{\varepsilon}{4}$
					and
					$\xi=\frac{\chi^2}{12}$\,;\label{step:P}

					Compute an estimate $\hat{E}\in \Comp$ such that 
					$\abs{\hat E-\bra{\psi} P(A')\ket{\psi}}\le \frac{\chi^2}{4}$\,;\label{step:est}

					\lIf{$\abs{\hat E}\ge \frac{\chi^2}{2}$}{output ``yes''\,;}

					\lElse{output ``no''\,;}
		\end{algorithm}
		\end{minipage}
	\end{center}

	Let $t^\ast$ be the smallest value of $t\in\set{0}{T-1}$ such that $\Test(t)$ outputs ``yes''. Define 
	\[
		E^\ast = t^\ast\frac{\varepsilon}{2}\upw-\upw.
	\]
	The following claim, whose proof is given after the proof of this theorem, guarantees that $E^\ast$ is a correct estimate of $\en{A}$.
	\begin{claim}\label{claim}
			$\abs{E^\ast-\en{A}}\le\varepsilon\,\upw$\,.
	\end{claim}	

	We now discuss the implementation of Step \ref{step:est}. For any $\delta\in(0,1]$, we describe how to compute an estimate $E^\ast$ such that \cref{eq:EVcond} holds with probability at least $1-\delta$. Taking $\delta=1/\exp(n)$ then proves the theorem. 
	
	We use the algorithm of \cref{th:PT} for the problem
	$\EST(s,\deg(P),\chi^2/4)$ in order to obtain
	an estimator $\hat E$ such that
	\[
		\abs{\hat E-\bra{\psi} P(A')\ket{\psi}}\le \chi^2/4
	\]
	holds with probability at least $1-\delta/T$. Since $\Test(t)$ is called at most $T$ times, the union bound guarantees that with probability at least $1-\delta$ no error occur during these tests. This implies that the output $E^\ast$ satisfies the bound of \cref{claim} with probability at least $1-\delta$.

	The overall time complexity is 
	\[
		\myO{T\cdot (s+1)^{c\cdot \deg(P)} \chi^{-8}}
		=
		\myO{s^{\frac{c' \log(1/\chi)}{\varepsilon}}}
		,
	\] 
	for some universal constant $c'>0$. Since $\deg(P)=\myO{1/\varepsilon}$,
	the space complexity is $\myO{1/\varepsilon^2}$.
\end{proof}
\begin{proof}[Proof of \cref{claim}]
	We first analyze the behavior of the procedure $\Test(t)$ 
	for any $t\in\set{0}{T-1}$. Let $P$ and $\tau,\theta,\xi$ be the polynomial and the parameters considered at Step \ref{step:P}. Let us write 
	\[
		A=\sum_{i=1}^N\lambda_i\ket{u_i}\bra{u_i}
	\]
	the eigenvalue decomposition of $A$, with $\lambda_1<\lambda_2<\cdots<\lambda_N$. The eigenvalue decomposition of $A'$ is 
	\[
		\sum_{i=1}^N \frac{\upw+\lambda_i}{2\upw}\ket{u_i}\bra{u_i}.
	\]
	We have  
	\[
		\bra{\psi}P(A')\ket{\psi}=\sum_{i=1}^N P\left(\frac{\upw+\lambda_i}{2\upw}\right)\abs{\braket{u_i}{\psi}}^2.
	\]
	We consider two cases. 
	\begin{itemize}
		\item[(a)] If $\en{A}\le (t-1)\frac{\varepsilon}{2}\upw-\upw$ then for any $i\in [N]$ such that $\lambda_i\le \en{A}+\frac{\varepsilon}{2}\upw$ we have 
		\[
			\frac{\upw+\lambda_i}{2\upw}\le \tau.
		\]
		From \cref{eq:P}, we thus obtain
		\[
			\bra{\psi} P(A')\ket{\psi}
			\ge 
			(1-\xi)\chi^2
			= 
			\left(1-\frac{\chi^2}{12}\right)\chi^2
			\ge\frac{11\chi^2}{12}
		\]
		and thus
		$
			\abs{\hat E}\ge \frac{11\chi^2}{12}-\frac{\chi^2}{4} \ge \frac{2\chi^2}{3}.
		$
		$\Test(t)$ thus outputs ``yes''.
		\item[(b)] If $\en{A} \ge (t+1)\frac{\varepsilon}{2}\upw-\upw$ then for all $i\in [N]$ we have 
		\[
			\frac{\upw+\lambda_i}{2\upw}\ge \tau+\theta.
		\]
		From \cref{eq:P}, we thus obtain
		\[
			\bra{\psi} P(A')\ket{\psi}
			\le\sum_{i=1}^N \xi \abs{\braket{u_i}{\psi}}^2
			\le \xi=\frac{\chi^2}{12}\,,
		\]
		which gives
			$\abs{\hat E}\le \frac{\chi^2}{12}+\frac{\chi^2}{4} = \frac{\chi^2}{3}$.
		$\Test(t)$ thus outputs ``no''. 
	\end{itemize}

	We are now ready to prove the claim.
	Assume that $\abs{E^\ast-\en{A}}>\varepsilon$. There are two cases to consider:
	\begin{itemize}
		\item If $E^\ast-\en{A}>\varepsilon\,\upw$, then 
		\[
			\en{A}<t^\ast\frac{\varepsilon}{2}\upw-\upw-\varepsilon\,\upw
		=(t^\ast-2)\frac{\varepsilon}{2}\upw-\upw.
		\]
		Note that since $\en{A}\in[-\upw,\upw]$, this can happen only if $t^\ast> 1$.
		From Case (a) of the above argument, $\Test(t^\ast-1)$ should output ``yes'', which contradicts the definition of~$t^\ast$.
		\item If $\en{A}-E^\ast>\varepsilon\, \upw$, then 
		\[
			\en{A}>t^\ast\frac{\varepsilon}{2}\upw+\varepsilon\,\upw-\upw=(t^\ast+2)\frac{\varepsilon}{2}\upw-\upw>(t^\ast+1)\frac{\varepsilon}{2}\upw-\upw.
		\]
		From  Case (b) of the above argument, $\Test(t^\ast)$ should thus output ``no'', which contradicts the definition of~$t^\ast$.
	\end{itemize}
	Since we get a contradiction in both cases, we conclude that $\abs{E^\ast-\en{A}}\le\varepsilon\,\upw$.
\end{proof}

\subsection{Consequences: Theorems \ref{th:main-informal} and \ref{cor:main-informal}}\label{sub:EV2}
We are now ready to give the full statements of Theorems \ref{th:main-informal} and \ref{cor:main-informal} and prove them.
\addtocounter{theorem}{-3}
\begin{theorem}[Full version]
	Consider any $\varepsilon,\chi\in(0,1]\cap\Real[\poly(n)]$.
	For any $k$-local Hamiltonian $H$ on $n$ qubits, given sample-and-query access to a quantum state $\ket{\psi}$ with 
	\[
		\overlapf{H,\ket{\psi}}{\frac{\varepsilon}{2}\kappa}\ge\chi,
	\] 
	where $\upw=\sum_{i=1}^m\norm{H_i}$,
	there is a classical algorithm that computes in $\poly\big(\frac{1}{\chi^{k/\varepsilon}},n\big)$ time and $\myO{\frac{1}{\varepsilon^2}}$ space an estimate $\hat E$ such that 
	\[
		\abs{\hat E-\en{H}}\le \varepsilon\sum_{i=1}^m\norm{H_i}
	\] 
	holds with probability at least $1-1/\exp(n)$.
\end{theorem}
\addtocounter{theorem}{+3}
\begin{proof}[Proof of \cref{th:main-informal}]
	We apply \cref{th:EV} with $A=H$, $s=2^k$ and $\kappa=\sum_{i=1}^m\norm{H_i}$.
\end{proof}
\addtocounter{theorem}{-3}
\begin{theorem}[Full version]
	Consider any $\varepsilon\in(0,1]\cap\Real[\poly(n)]$.
	For any $k$-local Hamiltonian $H$ on $n$ qubits, there is a classical algorithm that computes in $2^{O(kn/\varepsilon)}$ time and $\myO{\frac{1}{\varepsilon^2}}$ space an estimate $\hat E$ such that 
	\[
		\abs{\hat E-\en{H}}\le \varepsilon\sum_{i=1}^m\norm{H_i}
	\] 
	holds with probability at least $1-1/\exp(n)$.
\end{theorem}
\addtocounter{theorem}{+3}
\begin{proof}[Proof of \cref{cor:main-informal}]
	Let us write
	\[
		H=\sum_{i=1}^{2^n}\lambda_i\ket{u_i}\bra{u_i}
	\]
	the spectral decomposition of $H$, with $\lambda_1\le\lambda_2\le\cdots\le\lambda_{2^n}$ and corresponding orthonormal eigenvectors $\ket{u_1},\ldots,\ket{u_{2^n}}$,
	where $\lambda_1=\en{H}$.
	We apply \cref{th:main-informal} with the Hamiltonian $H'=H\otimes I$ acting on $2n$ qubits (here $I$ is the identity matrix acting on $n$ qubits) and guiding state
	\[
		\ket{\Phi}=\frac{1}{\sqrt{2^n}}\sum_{i=1}^{2^n}\ket{i}\ket{i},
	\]
	for which it is trivial to implement sample-and-query access. This is a maximally entangled state, which can also be written as 
	\[
		\ket{\Phi}=\frac{1}{\sqrt{2^n}}\sum_{i=1}^{2^n}\ket{u_i}\ket{v_i},
	\]
	for another orthonormal basis $\{\ket{v_1},\ldots,\ket{v_{2^n}}\}$.

	Let $t\in[2^n]$ denote the multiplicity of the ground state energy of $H$. The eigenspace corresponding to the ground state energy of $H'$ is thus
	$\spa\left\{\ket{u_i}\ket{j}\:|\:i\in[t],j\in[2^n]\right\}$.
	We have 
	\[
		\overlapf{H,\ket{\Phi}}{0}=\sqrt{\sum_{i=1}^{t}\sum_{j=1}^{2^n}\abs{(\bra{u_i}\bra{j})\ket{\Phi}}^2}\ge\frac{1}{\sqrt{2^{n}}}\sqrt{\sum_{j=1}^{2^n}\abs{\braket{j}{v_1}}^2}=\frac{1}{\sqrt{2^{n}}}.
	\] 

	The conclusion follows from \cref{th:main-informal} with $\chi=2^{-n/2}$.
\end{proof}

%% file: appendix.tex
\section{Proof of \cref{lemma:Tang}}\label{app}

	Let $p$ be the probability distribution from \cref{def:samplable}. Consider the following procedure.
	\begin{center}
		\begin{minipage}{14.2 cm} \vspace{2mm}
		\begin{algorithm}[H]
			\SetAlgorithmName{}{}{List of Algorithms}
			\nonl \hspace{-4mm}Procedure $\Ss$ 
	
			Sample one index $j\in\{1,\ldots,N\}$ according to the probability distribution $p$\,;
	
			Query $\braket{j}{\psi}$ and ${w}_j$\,;

			{\bf return} $\frac{ w_j}{\braket{j}{\psi}}$\,;
			\end{algorithm}
		\end{minipage}
	\end{center}
	Let $X$ denote the complex random variable corresponding to the output of this procedure.
	We calculate the expectation and variance of $X$:
	\begin{align*}
	\ex{X}
	&=
	\sum_{j=1}^N p(j)\frac{ w_j}{\braket{j}{\psi}}
	=
	\sum_{j=1}^N \abs{\braket{j}{\psi}}^2\frac{ w_j}{\braket{j}{\psi}}
	=
	\sum_{j=1}^N \braket{\psi}{j}w_j
	=
	\braket{\psi}{w}\\
	\var{X}
	&\le
	\ex{\abs{X}^2}
	=
	\sum_{j=1}^N 
	p (j)\abs{\frac{ w_j}{\braket{j}{\psi}}}^2
	=
	\norm{w}^2\,.
	\end{align*}
	Step 1 requires $\poly(n)$ time and space. Step 2 requires time $\qt{w}$ and space $\qs{w}$. At Step 3 we need to implement the division of an integer with binary expansion of length $\len{w}\le \qt{w}$ by an integer with binary expansion of length $\poly(n)$. 
	One execution of Procedure $\Ss$ can thus ne implemented in time $\myO{\qt{w}}$ and space $\qs{w}+\len{w}+\poly(n)$.\footnote{Note that the $\len{w}+\poly(n)$ part 
	can be exponentially improved by using space-efficient division \cite{HesseICALP01}. Hereafter we do not try to make such improvements since they will be negligible when considering our applications of \cref{lemma:Tang}.}
	
	We apply $t$ times Procedure $\Ss$, for some integer $t$ to be set later, each time getting a complex number $X_i$, and output the mean. Let \[Y=\frac{X_1+\cdots+X_t}{t}\] denote the corresponding complex random variable. Since the variables $X_1$, $\ldots$, $X_t$ are independent, we have
	\begin{align*}
	\ex{Y}&=\ex{X},\\
	\var{Y}&=\frac{1}{t^2}\left(\var{X_1}+\cdots+\var{X_t}\right) = \frac{\var{X}}{t}\le \frac{\norm{w}^2}{t}.
	\end{align*}
	By Chebyshev's inequality we obtain
	\[
	\Pr\left[\abs{Y-\ex{Y}}> \frac{\varepsilon}{\sqrt{2}}\norm{w}\right]
	\le 
	\frac{2\var{Y}}{\varepsilon^2\norm{w}^2}
	\le
	\frac{2}{\varepsilon^2 t}.
	\]
	Taking $t=\Theta(1/\varepsilon^2)$ guarantees that the above probability is at most $1/4$. The overall time and space complexities are 
	\[
		\myO{\qt{w}\,\varepsilon^{-2}}
	\]
	and 
	\[
		\qs{w}+\len{w}+\poly(n)+O(\log(1/\varepsilon))
	\,,\] 
	respectively.
	
	Using Lemma~\ref{lemma:powering}, by repeating this process $O(\log (1/\delta))$ times, we can compute an estimate $a$ such that  
	\[
		\Pr\left[\abs{a-\ex{Y}}> \varepsilon\norm{w}\right]
		\le 
		\delta.
	\]
	Lemma \ref{lemma:powering} introduces a $\log(1/\delta)$ factor in the time complexity, giving overall time complexity 
	\[
		\myO{\qt{w}\,\varepsilon^{-2}\log (1/\delta)}.
	\] 
	For the computation of the medians in Lemma \ref{lemma:powering}, we need to store $O(\log(1/\delta))$ values, each requiring $\len{w}+\poly(n)+O(\log(1/\varepsilon))$ bits.
	The overall space complexity is thus
	\[
		\qs{w}+\myO{\Big(\len{w}+\log(1/\varepsilon)\Big)\log(1/\delta)}.
	\]